\documentclass[12pt]{amsart}
\usepackage[margin=1in]{geometry}
\usepackage{amsfonts}
\usepackage{amsmath}
\usepackage{amssymb}
\usepackage{amsthm}
\usepackage{graphicx}
\usepackage{amsfonts}
\usepackage{color}

\newtheorem{theorem}{Theorem}[section]

\newtheorem{corollary}[theorem]{Corollary}

\newtheorem{remark}[theorem]{Remark}
\newtheorem{example}[theorem]{Example}
\newtheorem{definition}[theorem]{Definition}

\newcommand{\cA}{{\mathcal A}}

\newcommand{\R}{{\mathord{\mathbb R}}}
\newcommand{\C}{{\mathord{\mathbb C}}}
\newcommand{\Z}{{\mathord{\mathbb Z}}}
\newcommand{\N}{{\mathord{\mathbb N}}}
\newcommand{\E}{{\mathord{\mathbb E}}}
\newcommand{\PP}{{\mathord{\mathbb P}}}

\newcommand{\be}{\begin{equation}}
\newcommand{\ee}{\end{equation}}
\newcommand{\bea}{\begin{eqnarray}}
\newcommand{\eea}{\end{eqnarray}}

\def\idty{{\mathchoice {\mathrm{1\mskip-4mu l}} {\mathrm{1\mskip-4mu l}} %
{\mathrm{1\mskip-4.5mu l}} {\mathrm{1\mskip-5mu l}}}}

\numberwithin{equation}{section}

\begin{document}

\title[Disordered Quantum Spin Systems]{Dynamical Localization in\\ Disordered Quantum Spin Systems}

\author[E. Hamza]{Eman Hamza$^1$}
\address{$^1$ Department of Physics\\
Faculty of Science, Cairo University, Cairo 12613, Egypt}
\email{eman.hamza07@gmail.com}

\author[R. Sims]{Robert Sims$^2$}
\thanks{R.\ S.\ was supported in part by NSF grants DMS-0757424 and DMS-1101345}
\address{$^2$ Department of Mathematics\\
University of Arizona\\
Tucson, AZ 85721, USA}
\email{rsims@math.arizona.edu}

\author[G. Stolz]{G\"unter Stolz$^3$}
\thanks{G.\ S.\ was supported in part by NSF grants DMS-0653374 and DMS-1069320.}
\address{$^3$ Department of Mathematics\\
University of Alabama at Birmingham\\
Birmingham, AL 35294 USA}
\email{stolz@math.uab.edu}

\date{}
\maketitle

\vspace{.3truein} \centerline{\bf Abstract}

\medskip
{We say that a quantum spin system is dynamically localized if the time-evolution of 
local observables satisfies a zero-velocity Lieb-Robinson bound.
 In terms of this definition we have the following main results:
 First, for general systems with short range interactions, dynamical localization implies 
exponential decay of ground state correlations, up to an explicit correction. Second, the dynamical localization of random xy spin chains can be reduced to dynamical localization of an effective one-particle Hamiltonian. In particular, the isotropic xy chain in random exterior magnetic field is dynamically localized. }
%
%
%

\section{Introduction}

 One important reason for the popularity of quantum spin systems in theoretical and mathematical physics is that they can serve as relatively simple models to study interacting quantum many-body systems. In this setting, the physics of a single spin is governed by a finite-dimensional Hilbert space and hence (essentially) trivial. Therefore quantum spin systems are ideal models to focus directly and exclusively on physical phenomena arising from many-particle scattering. For the same reason one may want to study {\it random} quantum spin systems to better understand the effect of disorder on interacting quantum many-particle systems. This is our goal here.

Naturally, there is a considerable amount of physics literature on such models and we will not attempt to survey these works here (but we give some references to papers closely related to our work in Section~\ref{sec:isoxydynloc} below). Much less has been done on the rigorous mathematical level, where there was a short outburst of activity in the early 1990s, including the papers \cite{KleinPerez}, \cite{CKP}, \cite{KleinPerez92} and \cite{AKN}. This came at the end of a decade which had seen rapid progress in the theory of the Anderson model and, in particular, the development of several methods which allowed to give rigorous localization proofs. The referenced works from the early 90s built on the results and strategies which arose from studying the Anderson model and found ways to apply some of these ideas to quantum spin systems with disorder.

Our work here arose from a desire to return to these investigations from the point of view of progress made and new questions having been asked over the last two decades. One important development in this period was the increased attention to questions of quantum {\it dynamics} and, in particular, the development of new tools and refinement of existing tools to prove dynamical localization properties for the Anderson model and other types of random Schr\"odinger operators. A concrete goal of the present work is to discuss the meaning of dynamical localization in the context of quantum spin systems and to rigorously prove this for at least one class of models. Let us say from the outset that our model of choice, the xy spin chain in a random exterior magnetic field, is particularly simple as its localization properties can be directly reduced to those of the Anderson model. A future goal should be to understand dynamical localization properties for other models of disordered quantum spin systems, including multi-dimensional systems.

Another reason for returning to the study of disordered quantum spin systems is recent progress in rigorously understanding localization for the Anderson model with electron-electron interactions. The papers \cite{CS1,CS2} and \cite{AW} showed how to extend the methods of multiscale analysis and fractional moments, respectively, to obtain localization results in suitable regimes for a system of $N$ interacting electrons in a random potential. These methods do not yet allow to study the thermodynamic limit of an electron gas in a random environment, i.e.\ an infinite volume limit in which the number of electrons is proportional to the volume. In the physics literature, e.g.\ \cite{Baskoetal}, it has been suggested that localization will persist in the case of low electron density. For a chance to investigate this rigorously it is necessary to grasp the consequences of many-particle scattering effects in such systems. Characterizing localization in terms of the properties of the dynamics, rather than the spectrum, seems to be better suited to these many-body systems. Interacting Anderson models, due to the infinite-dimensional single particle dynamics, are much harder to analyze than random quantum spin systems. Nevertheless, by starting with the latter one can hope to get insights on how many-body localization can arise.

The dynamics of short range spin systems satisfy a locality estimate commonly known as {\it Lieb-Robinson bounds} \cite{LR1}. 
These estimates can be understood as upper bounds on the group velocity of spin waves or, alternatively, as bounds on information propagation in a spin system
 (a term inspired by quantum information theory, where quantum spin systems are thought of as models for interacting qubits). As reviewed in Section~\ref{Sec:cordec} below, Lieb-Robinson bounds can be proven for very general classes of quantum spin systems and have recently found much renewed interest and many applications.

We will define a strong form of dynamical localization for a quantum spin system in terms of the validity of a {\it zero-velocity Lieb-Robinson bound}, e.g.\ Definition~\ref{def:0lrv}. Given this, our two main results can be summarized as follows:

(i) If a zero-velocity Lieb-Robinson bound holds in a given quantum spin system, then --- up to a logarithmic correction --- the system satisfies exponential clustering, i.e.\ ground state correlations decay exponentially, independent of the system's size. This is the content of Theorem~\ref{thm:expdec}, stated and proven in Section~\ref{subsec:cordecay} below. Exponential clustering is known for quantum spin systems with a size-independent ground state energy gap. Thus a zero-velocity Lieb-Robinson bound can substitute for a ground state gap in the proof of exponential clustering.

(ii) We will prove that the isotropic xy spin chain in an exterior random magnetic field satisfies a zero-velocity Lieb-Robinson bound after averaging over the disorder. This will be stated as Corollary~\ref{cor:0vLRisoxy}. More generally, we review in Section~\ref{subsec:diag} the well known reduction (e.g.\ \cite{LSM}) of the {\it anisotropic} xy chain to an effective single particle Hamiltonian. We then show in Theorem~\ref{thm:main} that dynamical localization of the single particle Hamiltonian implies a zero-velocity Lieb-Robinson bound for the anisotropic xy chain. In the isotropic case the single particle Hamiltonian is given by the Anderson model, where the corresponding dynamical localization property is well known.

In Section~\ref{subsec:isocordecay} we combine these two results to conclude exponential clustering for the isotropic random xy chain, see Theorem~\ref{thm:isoxycordecay}. This requires some modification of Theorem~\ref{thm:expdec} and its proof, due to the fact that Theorem~\ref{thm:isoxycordecay} is not deterministic but refers to expected values of the relevant quantities. This result is closely related to earlier work in \cite{KleinPerez} where a direct proof of correlation decay for the isotropic random xy chain is given. In contrast to our result, the proof in \cite{KleinPerez} does not proceed via dynamical localization. We will further discuss this, as well as other related work in e.g. \cite{KMM}, within the concluding remarks of Section~\ref{sec:conrem}.

Let us finish this introduction by stressing what we consider to be the most important novel contributions of our work: 
(i) We propose the new concept of a {\it zero-velocity Lieb-Robinson bound} as a means to describe dynamical localization of many-body quantum systems. 
Our definition, which we can verify holds for certain models, is a simplified version of the {\it mobility gap} proposed by Hastings in \cite{Hastings}, see 
Section~\ref{subsec:cordecay} for more discussion. In particular, we provide a general argument which shows that this form of dynamical 
localization implies exponential decay of ground state correlations. (ii) As a first example of a many-body system where a 
zero-velocity Lieb-Robinson bound can be rigorously established we study disordered xy chains. The core argument is to 
reduce the analysis, via a Jordan-Wigner transform, to dynamical localization for Anderson-type single particle Hamiltonians. 
One obstacle in proving our estimates, which we stress are uniform in time, is to overcome the non-locality of the Jordan-Wigner transform, and
this is accomplished in Section~\ref{subsec:dl0v}. This argument goes beyond previously known non-rigorous results in the physics literature (see Section~\ref{sec:isoxydynloc} for a discussion of the results in \cite{BO}).

\section{Lieb-Robinson Bounds and Correlation Decay} \label{Sec:cordec}

%
%
%

The goal of this section is to prove an estimate on the decay of correlations
in systems that satisfy a strong form of the Lieb-Robinson bound, see
Definition~\ref{def:0lrv} below. Although we state and prove this bound for deterministic
systems, the most interesting models for which these bounds have
been established are random, more on this in Sections~\ref{sec:xychain} and \ref{sec:isoxychain}.
We begin this section by introducing a general class of quantum spin systems
to which our result applies. Next, we briefly discuss Lieb-Robinson bounds
and a result on exponential clustering for gapped systems. Finally, we define what we mean by a zero-velocity Lieb-Robinson bound and then state and
prove our main result on correlation decay in systems which are dynamically localized in this sense, Theorem~\ref{thm:expdec} below.

\subsection{Quantum Spin Systems} \label{Sec:qss}

Our results on correlation decay apply to a very general class of multi-dimensional quantum spin systems. In fact, while we choose the set of spin sites to be $\Z^{\nu}$, with some more notational effort we could equally well consider systems over much more general graphs. We feel, however, that this would not be beneficial to our presentation.

A quantum spin system over $\mathbb{Z}^{\nu}$ is defined as follows.
To each $m \in \mathbb{Z}^{\nu}$, associate a finite dimensional Hilbert Space
$\mathcal{H}_m = \mathbb{C}^{n_m}$. $\mathcal{H}_m$ is commonly referred to as
the single-site Hilbert space, and the dimension $n_m \geq 2$ is related
to the spin $J_m$ at site $m$ by $n_m = 2 J_m +1$, e.g., $n_m =2$ corresponds to spin
$J_m = 1/2$. For any finite set $N \subset \mathbb{Z}^{\nu}$, a composite
Hilbert space and algebra of observables are defined by setting
\[
\mathcal{H}_{N} = \bigotimes_{m \in N} \mathcal{H}_m \quad \mbox{and} \quad \mathcal{A}_{N} = \bigotimes_{m \in N} \mathcal{B}( \mathcal{H}_m) \, ,
\]
where $\mathcal{B}( \mathcal{H}_m)$ denotes the bounded linear operators over $\mathcal{H}_m$, i.e.,
the set of $n_m \times n_m$ complex matrices. Due to the tensor product structure, it is
clear that for any finite sets $N_0 \subset N \subset \mathbb{Z}^{\nu}$,
each observable $A \in \mathcal{A}_{N_0}$ can be identified with an
observable $A' = A \otimes \idty_{N \setminus N_0} \in \mathcal{A}_{N}$.
In this case, we regard $\mathcal{A}_{N_0} \subset \mathcal{A}_{N}$ and thereby
define inductively the algebra of all local observables
\[
\mathcal{A}_{\rm loc} = \bigcup_{N \in \mathcal{P}_0( \mathbb{Z}^{\nu})} \mathcal{A}_{N}
\]
where the union is taken over all finite subsets of $\mathbb{Z}^{\nu}$.

A model on such a quantum spin system is defined through an interaction.
An interaction is a mapping $\Phi : \mathcal{P}_0( \mathbb{Z}^{\nu}) \to \mathcal{A}_{\rm loc}$
satisfying $\Phi(M)^* = \Phi(M) \in \mathcal{A}_M$ for all finite $M \subset \mathbb{Z}^{\nu}$.
Corresponding to any interaction, there is an associated family of
local Hamiltonians, parametrized by finite subsets $N \subset \mathbb{Z}^{\nu}$, given
by
\begin{equation} \label{eq:locham}
H_{N}^{\Phi} = \sum_{M \subset N} \Phi (M) \, .
\end{equation}
When the interaction is fixed, we will often drop the dependence of the
local Hamiltonian on $\Phi$. Since the sum in (\ref{eq:locham}) is finite, $H_{N}$ is
self-adjoint. By the spectral theorem, there is a one-parameter group of
automorphisms, which we denote by $\{ \tau_t^{N} \}_{t \in \mathbb{R}}$, defined
by setting
\[
\tau_t^{N} (A) = e^{itH_{N}} A e^{-itH_{N}} \qquad \text{for all } A \in \cA_{N} \,.
\]
$\tau_t^{N}$ is called the Heisenberg dynamics or time evolution corresponding to
$H_{N}$.

The above general framework covers a wealth of heavily studied explicit models, such as the Heisenberg, Ising, XY and XXZ models, all describing spin 1/2 particles. An important mathematical model for higher spin is the spin 1 AKLT model \cite{AKLT}. For general background on the mathematical treatment of quantum spin systems as models of statistical mechanics we refer to \cite{BratRob} and, for a different perspective, \cite{Simon}.

As an example we provide the spin 1/2 Heisenberg model, probably the oldest and most prominent example.

\begin{example} A common model is the Heisenberg Hamiltonian.
In this case, one takes as single-site Hilbert space $\mathcal{H}_m = \mathbb{C}^2$
for all $m \in \mathbb{Z}^{\nu}$. The interaction for this model is given by
\[
\Phi(M) = \left\{ \begin{array}{cc} \mu \left( \sigma_n^x \sigma_m^x+ \sigma_n^y \sigma_m^y + \sigma_n^z \sigma_m^z \right) & \mbox{if } M = \{ n, m\} \mbox{ and } |n-m|=1, \\
0 & \mbox{otherwise},  \end{array} \right.
\]
where the real number $\mu$ is the parameter of the model,
\begin{equation}\label{eq:pauli}
\sigma^x = \left( \begin{array}{cc} 0 & 1 \\ 1 & 0  \end{array} \right) , \quad \sigma^y = \left( \begin{array}{cc} 0 & -i \\ i & 0  \end{array} \right) ,  \quad \mbox{and} \quad \sigma^z = \left( \begin{array}{cc} 1 & 0 \\ 0 & -1  \end{array} \right) \,
\end{equation}
are the Pauli-spin matrices, and for any $w \in \{x, y, z\}$, $\sigma_n^w \in \mathcal{A}_M$ is the matrix $\sigma^w \otimes \idty$ with $\sigma^w$ in the $n$-th factor.
\end{example}

\subsection{Some Prior Results}
In 1972, Lieb and Robinson proved a locality estimate for a large class of quantum spin models \cite{LR1}.
Their observation can be described as follows. Consider a quantum spin system over
$\mathbb{Z}^{\nu}$, and let $J$ and $K$ be finite, disjoint subsets of
$\mathbb{Z}^{\nu}$. Take any finite subset $N$ with $J \cup K \subset N$.
It is clear from the tensor product structure of the observable algebras that for any
$A \in \mathcal{A}_J$ and $B \in \mathcal{A}_K$, $[A \otimes \idty_{N \setminus J}, B \otimes \idty_{N \setminus K}]=0$.
To ease notation, we will often suppress these identities and simply regard $A$ and $B$ as observables
in $\mathcal{A}_{N}$ whenever $J,K \subset N$. For a large class of essentially short-range interactions, Lieb and
Robinson proved a bound of the form
\begin{equation} \label{eq:lrb}
\left\| \left[ \tau_t^{N}(A), B \right] \right\| \leq C(J,K) \|A\| \|B\| e^{- \eta \left( d(J,K) - v|t| \right)} \, ,
\end{equation}
where $d(J,K) := \min\{|j-k|: j \in J, k \in K\}$.

This estimate shows that for times $t$ with $|t| \leq d(J,K) /v$ the commutator remains exponentially small, and
so disturbances do not propagate through the system, by the Heisenberg dynamics, arbitrarily fast.
It is important to note that the numbers $C(J,K)$, $\eta$, and $v$ are all independent of the volume $N$
on which the dynamics is defined. The number $v$, which can be made explicit in terms of an appropriate norm
on the interaction, is called a bound on the Lieb-Robinson velocity of the model under consideration.
Recently there have been a number of generalizations of these Lieb-Robinson bounds
\cite{nach12006, hast2006, nach22006, eisert2008, amour2009, nach2009, PS09, Pou10, nach2011}, and they
have proven useful in a variety of applications
\cite{hastings2004a, NS2, hastings2007, HS, nach2010, bravyi:2010a, bravyi:2010b, bachmann2011}.
Many of the previously mentioned results were reviewed in e.g. \cite{NS3, hastings2010, sims2010}.

One of the first applications of the new Lieb-Robinson bounds was a proof of the Exponential Clustering
theorem \cite{nach12006, hast2006}. Stated simply, exponential clustering describes the fact that ground-state expectations of
gapped quantum spin systems decay exponentially in space. As one might imagine, the rate of exponential
decay depends on the size of the gap, see \cite{nach12006} for an explicit estimate and \cite{NS3.5} for an improvement.
It is clear that local Hamiltonians corresponding to any interaction will be gapped. The important
point in the proof of clustering is that if the model under consideration has a uniform positive lower bound
on the gaps in the finite volume, then the rate of decay persists in the thermodynamic limit.

Despite the fact that ground state
expectations are time-independent, the proof of this clustering result uses Lieb-Robinson bounds in a crucial
way. In fact, the proof requires a slightly stronger bound than the one claimed in (\ref{eq:lrb}):
\[
\left\| \left[ \tau_t^{N}(A), B \right] \right\| \leq |t| C(J,K) \|A\| \|B\| e^{- \eta (d(J,K)-v|t|)} \, ,
\]
at least for small times, e.g., $|t| \leq 1$.

\subsection{Correlation Decay in Dynamically Localized Systems} \label{subsec:cordecay}

Our first result in this work, see Theorem~\ref{thm:expdec} below, is a proof of a clustering-type
bound under different assumptions than we described above. In words, we prove that a {\it zero-velocity Lieb-Robinson bound}
implies exponential decay of correlations, up to a logarithmic correction. Interestingly, our method of
proof, motivated by some recent observations by Hastings in \cite{Hastings}, provides a decay rate independent
of the size of the gap, and our bound depends on the gap only through the logarithmic correction.

Before proceeding with this result, let us relate our work to the contents of \cite{Hastings}. There the concept of a {\it mobility gap} is proposed, which is considered a many-body version of dynamical localization (and inspired by the {\it mobility edge} concept for single particle systems). It is discussed that systems with a mobility gap have properties similar to gapped systems, including decay of ground state correlations. They are also said to satisfy a higher-dimensional Lieb-Schultz-Mattis theorem and Hall conductance quantification. A zero-velocity Lieb-Robinson bound in our sense can be considered a special case, and particularly strong form, of a mobility gap. This allows for a simplified argument, although with methods similar to those of \cite{Hastings}, in our derivation of correlation decay below.

The goal of this section is to state and prove Theorem~\ref{thm:expdec} below.
We begin with the following definition, which provides a way to characterize dynamical localization of a quantum spin system.

\begin{definition} \label{def:0lrv} Let $\Phi$ be an interaction on a quantum spin system over $\mathbb{Z}^{\nu}$.
We say that $\Phi$ satisfies a zero-velocity Lieb-Robinson bound if there exists $\eta >0$ with the following property: Given any finite, disjoint subsets $J, K \subset \Z^{\nu}$, there exists $C(J,K) <\infty$ such that
\begin{equation} \label{eq:0lrv}
\left\| \left[ \tau_t^{N}(A), B \right] \right\| \leq C(J,K)  \min[|t|, 1]  \| A \| \| B \| e^{- \eta d(J,K)}
\end{equation}
for all finite $N\subset \Z^{\nu}$ with $J \cup K \subset N$, all $t\in \R$ and all $A\in \mathcal{A}_J$ and $B\in \mathcal{A}_K$.
\end{definition}

Here it is important that the numbers $\eta$ and $C(J,K)$ do not depend on the size of $N$ on which the
dynamics is defined. In applications one will usually be able to say more about the dependence of $C(J,K)$ on $J$ and $K$, see, e.g., Section~\ref{sec:xychain}. For bounds of the form (\ref{eq:0lrv}) to be useful, $C(J,K)$ should only depend on simple geometric properties of $J$ and $K$, such as their size or surface area.

Also note here that $\|[\tau_t^N(A),B]\| = \|[A, \tau_{-t}^N(B)]\|$, so that bounds such as (\ref{eq:0lrv}) can be applied equally well to $\|[A, \tau_t^N(B)]\|$.

\begin{theorem} \label{thm:expdec} Let $\Phi$ be an interaction on $\mathbb{Z}^{\nu}$ that satisfies a zero-velocity
Lieb-Robinson bound. Let $J$ and $K$ be finite, disjoint subsets of $\mathbb{Z}^{\nu}$ and
take a finite set $N \subset \mathbb{Z}^{\nu}$ with $J \cup K \subset N$.
Suppose the local Hamiltonian, $H_{N}$, corresponding to $\Phi$ satisfies $\min \sigma(H_N)=0$.
Let $\psi_0 = \psi_0( N)$ denote a normalized ground state, i.e., $H_{N} \psi_0 = 0$ and
$\| \psi_0 \| =1$ and let   $P_0$ is the orthogonal projection onto ${\rm ker}(H_{N})$. Then for all $A \in \mathcal{A}_J$ and $B \in \mathcal{A}_K$
with $P_0 B \psi_0 = P_0 B^* \psi_0 = 0$, the bound
\begin{equation} \label{eq:expdec}
\left| \langle \psi_0, A B \psi_0 \rangle \right| \leq \left[ 1 + \frac{C(J,K)}{\pi} \left( 2 - \ln \frac{\gamma}{\sqrt{\pi \eta d(J,K)}} \right)\right] \|A\| \|B\| e^{- \eta d(J,K) } \, ,
\end{equation}
holds. Here $\eta$ and $C(J,K)$ are as in Definition~\ref{def:0lrv},
$\gamma = \gamma_{N}$ is the spectral gap, i.e.
\[
\gamma = \min (\sigma(H_N) \setminus \{0\}) \, .
\]
\end{theorem}

Before we begin the proof of this theorem, some comments are in order. First, if the
ground state of $H_{N}$ is non-degenerate, then the additional assumption on
the observable $B$ is equivalent to $\langle \psi_0, B \psi_0 \rangle = 0$. In this case,
our result then demonstrates exponential decay of correlations $\left| \langle \psi_0, A B \psi_0 \rangle - \langle \psi_0, A \psi_0 \rangle \langle \psi_0, B \psi_0 \rangle\right|$ in the distance of $J$ and $K$ (apply Theorem~\ref{thm:expdec} to $\tilde{B} := B - \langle \psi_0, B\psi_0 \rangle \idty$).

Also, the assumption that $H_N$ has ground state energy $0$ is inessential and made exclusively for convenience in proofs. For more general $H_N$ with ground state energy $E_0^{(N)}$ one gets the same result by first applying Theorem~\ref{thm:expdec} to $\hat{H}_N := H_N - E_0^{(N)}$, which has the same ground state and Heisenberg evolution as $H_N$.

Next, for general systems and especially large sets $N$, we expect that $\gamma$ is quite small. In this case the meaning of (\ref{eq:expdec}) is
\[ \left| \langle \psi_0, A B \psi_0 \rangle \right| \lesssim \|A\| \|B\| |\ln \gamma| e^{-(\eta-)d(J,K)}, \]
in the usual sense that $\eta-$ can be $\eta-\varepsilon$ for and $\varepsilon>0$ and $\lesssim$ allows for an $\varepsilon$-dependent (and $(J,K)$-dependent) constant factor. Clearly, the rate of decay is independent of the gap size; up to the logarithmic correction in $\gamma$. The usefulness of this estimate depends on the context of the application and, in particular, how rapidly $\gamma$ approaches zero for large volume. We will discuss this for the example of the isotropic random xy chain in Section~\ref{sec:isoxychain} below.

Finally, similarly to \cite{nach12006} and \cite{hast2006}, see also \cite{Hastings}, our method of proof employs a specific approximation
technique. Fix finite sets $K \subset N \subset \mathbb{Z}^{\nu}$. For any local observable $B \in \mathcal{A}_K$, $\alpha >0$, and $\epsilon>0$,
a quasi-local approximation of $B$ is defined by setting
\[
B( \alpha, \epsilon) = \frac{1}{2 \pi i} \int_{\mathbb{R}} \tau_t^{N}(B) \frac{e^{- \alpha t^2}}{t-i \epsilon} \, dt \, .
\]
If $K \subset N$, then for each positive choice of $\alpha$ and $\epsilon$, the support
of $B( \alpha, \epsilon)$ is $N$, however, the explicit gaussian kernel enables useful
estimates. In fact, the following is a basic fact about Fourier transforms of gaussians which we
will use in our proof. Let $E \in \mathbb{R}$, $\alpha >0$, and $\epsilon >0$. 
\begin{equation} \label{eq:bigbad}
 \frac{1}{2 \pi i}
\, \int_{\mathbb{R}} \, \frac{e^{iEt} \, e^{- \alpha t^2}}{t - i \epsilon} \, dt \, = \, \frac{ 1}{2 \sqrt{ \pi \alpha}} \,
\int_0^{ \infty} \, e^{- \epsilon w} \, e^{ - \frac{(w-E)^2}{4 \alpha}} \, dw.
\end{equation}
This follows, e.g., from the identity
\begin{equation}
(\epsilon + i t)^{-1} = \int_0^{\infty} e^{-w(\epsilon+it)} \, dw \, .
\end{equation}

\begin{proof}[Proof of Theorem~\ref{thm:expdec}]

We begin by observing that
\begin{equation}\label{eq:corr}
\langle \psi_0, AB \psi_0 \rangle = \left\langle \psi_0, A \left(B - B( \alpha, \epsilon) \right) \psi_0 \right\rangle + \langle \psi_0, B( \alpha, \epsilon) A \psi_0 \rangle + \langle \psi_0, [A, B(\alpha, \epsilon)]  \psi_0 \rangle \,
\end{equation}
for any choices of positive $\alpha$ and $\epsilon$. An application of the spectral theorem and (\ref{eq:bigbad}) shows that
for any $x \in \mathcal{H}_{N}$,
\begin{eqnarray*}
\left\langle x,  B( \alpha, \epsilon)  \psi_0 \right\rangle  & = &
 \int_{0}^{\infty}  \frac{1}{2 \pi i} \int_{\mathbb{R}} \frac{e^{itE} e^{-  \alpha t^2}}{t-i \epsilon} \, dt \, d \langle x,  P_E B\psi_0 \rangle \nonumber \\
& = & \int_{\gamma}^{\infty}  \frac{1}{2 \sqrt{\pi \alpha}} \int_0^{\infty} e^{- \epsilon w} e^{- \frac{(w-E)^2}{4 \alpha}} \, dw  \, d \langle x,  P_E B\psi_0 \rangle \, ,
\end{eqnarray*}
where, for the last equality, we used the assumption that $P_0B \psi_0 = 0$. The integrand can be
re-written as
\begin{eqnarray*}
\lefteqn{\frac{1}{2 \sqrt{\pi \alpha}} \int_0^{\infty} e^{- \epsilon w} e^{- \frac{(w-E)^2}{4 \alpha}} \, dw} \nonumber \\
& = & 1 - \frac{1}{2 \sqrt{\pi \alpha}}  \int_{- \infty}^0 e^{- \frac{(w-E)^2}{4 \alpha}} \, dw + \frac{1}{2 \sqrt{\pi \alpha}} \int_0^{\infty} \left( e^{- \epsilon w} -1 \right) e^{- \frac{(w-E)^2}{4 \alpha}} \, dw \, \nonumber \\
& =: & 1 - R_1(E, \alpha) + R_2(E, \alpha, \epsilon) \, .
\end{eqnarray*}
Since $\limsup_{\epsilon \to 0}R_2(E, \alpha, \epsilon) = 0$ and  $2 R_1(E, \alpha) \leq e^{- \frac{\gamma^2}{4 \alpha}}$ for $E \geq \gamma$, it is clear that
\begin{eqnarray*}
\limsup_{\epsilon \to 0} \left| \left\langle \psi_0, A \left(B - B( \alpha, \epsilon) \right) \psi_0 \right\rangle \right|
& \leq & \int_{\gamma}^{\infty} R_1(E, \alpha) d \langle A^*\psi_0, P_E B \psi_0 \rangle \nonumber \\
& \leq & \frac{1}{2} e^{- \frac{\gamma^2}{4 \alpha}} \| A^* \psi_0 \| \| B \psi_0 \| \, .
\end{eqnarray*}

Similarly,
\begin{eqnarray*}
\langle \psi_0, B( \alpha, \epsilon) A \psi_0  \rangle & = & \int_{0}^{\infty}  \frac{1}{2 \pi i} \int_{\mathbb{R}} \frac{e^{-itE} e^{-  \alpha t^2}}{t-i \epsilon} \, dt \, d \langle B^* \psi_0,  P_E A \psi_0 \rangle   \nonumber \\
& = & \int_{\gamma}^{\infty}  R_3(E, \alpha, \epsilon)  \, d \langle B^* \psi_0,  P_E A \psi_0 \rangle   \, ,
\end{eqnarray*}
where
\[
R_3(E, \alpha, \epsilon) = \frac{1}{2 \sqrt{\pi \alpha}} \int_0^{\infty} e^{- \epsilon w} e^{- \frac{(w+E)^2}{4 \alpha}} \, dw \leq \frac{1}{2} e^{- \frac{E^2}{4 \alpha}} \, .
\]
The bound
\[
\limsup_{\epsilon \to 0} \left| \langle \psi_0, B( \alpha, \epsilon) A \psi_0  \rangle \right| \leq \frac{1}{2} e^{- \frac{\gamma^2}{4 \alpha}} \| B^* \psi_0 \| \| A \psi_0 \| \, ,
\]
readily follows.

To the final term, we apply the assumption of a zero-velocity Lieb-Robinson bound.
Clearly,
\begin{equation} \label{eq:thirdterm}
\left| \langle \psi_0, [A, B(\alpha, \epsilon)]  \psi_0 \rangle \right| \leq \left\| [A, B(\alpha, \epsilon)] \right\| \leq \frac{1}{2 \pi} \int_{\mathbb{R}} \frac{ \left\| \left[ A, \tau_t^{N}(B) \right] \right\|}{|t|}  e^{- \alpha t^2} \, dt \, .
\end{equation}
We divide the integral above into three pieces. Using (\ref{eq:0lrv}), it is clear that
\[
\frac{1}{2 \pi} \int_{|t| \leq 1} \frac{ \left\| \left[ \tau_t^N(A), B \right] \right\|}{|t|}  e^{- \alpha t^2} \, dt  \leq  \frac{C(J,K)}{\pi} \|A\| \|B\| e^{- \eta d(J,K)}  \, ,
\]
and similarly, for any $\lambda >1$,
\[
\frac{1}{2 \pi} \int_{1\leq |t| \leq  \lambda} \frac{ \left\| \left[ \tau_t^N(A), B \right] \right\|}{|t|}  e^{- \alpha t^2} \, dt  \leq  \frac{C(J,K)}{\pi} \|A\| \|B\|
e^{- \eta d(J,K)} \ln( \lambda ) .
\]
Lastly,
\[
\frac{1}{2 \pi} \int_{ |t|> \lambda} \frac{ \left\| \left[ \tau_t^N(A), B \right] \right\|}{|t|}  e^{- \alpha t^2} \, dt  \leq  \frac{C(J,K)}{2\pi}\|A\| \|B\| e^{- \eta d(J,K)}
\frac{1}{\lambda} \sqrt{\frac{\pi}{\alpha}} .
\]

Now with the choice of
\begin{equation} \label{eq:alphalambda}
\alpha = \frac{\gamma^2}{4 \eta d(J,K)} \quad \mbox{and} \quad 2 \lambda = \sqrt{ \frac{\pi}{\alpha}} \, ,
\end{equation}
we have proven that
\[
\left| \langle \psi_0, AB \psi_0 \rangle \right| \leq \| A \| \| B \| e^{- \eta d(J,K)} \left( 1 + \frac{C(J,K)}{\pi} \left( 2 + \ln  \sqrt{ \frac{ \pi}{4 \alpha}}  \right) \right) \, ,
\]
as claimed.
\end{proof}

%
%

\section{The $xy$-chain with disorder} \label{sec:xychain}

In this section we consider the anisotropic xy-chain with free boundary conditions, as introduced
in \cite{LSM}. We start by reviewing the diagonalization of this xy Hamiltonian with
deterministic coefficients. In essence, this consists of using a Jordan-Wigner transform to reduce the $n$-body xy Hamiltonian to a free Fermion system whose diagonalization is governed by the diagonalization of an effective one-body Hamiltonian. The latter takes the form of a $2\times 2$-block Jacobi matrix.

Next, we consider the case where the coefficients in the xy chain are random
variables. In this case the corresponding block-Jacobi matrix is a random operator closely related to the Anderson model.
We demonstrate that if the block-Jacobi matrix is dynamically
localized, then the Heisenberg dynamics corresponding to the xy Hamiltonian
satisfies a zero-velocity Lieb-Robinson bound in average.

\subsection{Diagonalizing the xy Chain} \label{subsec:diag}
Consider three real-valued sequences $\{ \mu_j \}$, $\{ \gamma_j \}$, and $\{ \nu_j \}$.
These parameters will represent the coupling strength, the anisotropy, and the external magnetic field
respectively. We will assume $\mu_j \not= 0$ for all $j$ as otherwise the chain decomposes into shorter pieces. To each integer $n \geq 1$, we will denote by $H_{[1,n]}$ the finite volume, anisotropic xy Hamiltonian with
free boundary conditions given by
\begin{equation} \label{eq:anisoxychain}
H_{[1,n]} = \sum_{j=1}^{n-1} \mu_j [ (1+\gamma_j) \sigma_j^x \sigma_{j+1}^x + (1-\gamma_j) \sigma_j^y \sigma_{j+1}^y] + \sum_{j=1}^n \nu_j \sigma_j^z \, .
\end{equation}
Here $\sigma^x,\sigma^y$, and $\sigma^z$ are the Pauli matrices given in \eqref{eq:pauli}, and
$H_{[1,n]}$ acts on the Hilbert space $\mathcal{H}_{[1,n]} = \bigotimes_{j=1}^n \C^2$. For notational
convenience, we will write $H_n = H_{[1,n]}$ below.

It is well-known, see e.g. \cite{LSM}, that Hamiltonians of the form (\ref{eq:anisoxychain}) can be
diagonalized in terms of a system of operators satisfying canonical anti-commutation relations (CAR).
To introduce useful notation, we briefly discuss this diagonalization procedure below.

First, one introduces raising and lowering operators as
\[
a_j^* = \frac{1}{2}(\sigma_j^x +i \sigma_j^y)  \quad \mbox{and} \quad a_j = \frac{1}{2}(\sigma_j^x -i \sigma_j^y) \quad \mbox{for each } 1 \leq j \leq n.
\]
Using the relations
\[
\sigma_j^x \sigma_{j+1}^x \pm \sigma_j^y \sigma_{j+1}^y = \left\{ \begin{array}{cc} 2 \left( a_j^* a_{j+1} + a_{j+1}^* a_j \right) & \mbox{if } +, \\ 2 \left( a_j a_{j+1} + a_{j+1}^* a_j^* \right)  & \mbox{if } -, \end{array} \right. \quad \mbox{and} \quad \sigma_j^z = 2 a_j^* a_j - \idty \, ,
\]
one sees that
\begin{equation} \label{eq:hamas}
H_n  =  2 \sum_{j=1}^{n-1} \mu_j[a^*_j a_{j+1} + a_{j+1}^* a_j + \gamma_j (a_j a_{j+1} + a_{j+1}^* a_j^*)] + \sum_{j=1}^n\nu_j (2a_j^* a_j- \idty).
\end{equation}
Since the $a$-operators preserve the local structure, they do not satisfy anti-commutation relations.
To break this locality, one introduces new variables via the Jordan-Wigner transformation. Set
\begin{equation} \label{eq:cdef}
c_1 = a_1 \quad \mbox{and} \quad c_j = \sigma_1^z \cdots \sigma_{j-1}^z a_j \quad \mbox{for all } 2 \leq j \leq n.
\end{equation}
A short calculation shows that
\[
c_j^* c_j = a_j^* a_j, \quad c_j^*c_{j+1} = - a_j^* a_{j+1}, \quad \mbox{and} \quad c_jc_{j+1} = a_ja_{j+1} \, ,
\]
in which case we find that
\begin{equation} \label{eq:hamcs}
H_n  =  2 \sum_{j=1}^{n-1} \mu_j[-c^*_j c_{j+1} - c_{j+1}^* c_j + \gamma_j (c_j c_{j+1} + c_{j+1}^* c_j^*)] + \sum_{j=1}^n\nu_j (2 c_j^* c_j- \idty).
\end{equation}

In contrast to the $a$-operators, the $c$-operators do satisfy the CAR. In fact,
one easily calculates that for all $1 \leq j,k \leq n$,
\begin{equation} \label{eq:car}
\begin{array}{cc}
\left\{ c_j, c_k^* \right\} = \delta_{j,k} \idty \, , \\
\left\{ c_j, c_k \right\}  = \left\{ c_j^*, c_k^* \right\}  = 0 \, ,
\end{array}
\end{equation}
which are the canonical anti-commutation relations. It is convenient, and equivalent, to express these
CAR through a vector-valued formalism. Let  $\mathcal{C} = (c_1, \ldots, c_n, c_1^*, \ldots, c_n^*)^t$, then
it is clear that the $c$-operators satisfy (\ref{eq:car}) if and only if
\begin{equation} \label{eq:carvec}
\mathcal{C}\mathcal{C}^* + \mathcal{J}(\mathcal{C}\mathcal{C}^*)^t \mathcal{J} = \idty \, ,
\end{equation}
where
\[
 \mathcal{J} = \left( \begin{array}{cc} 0 & \idty \\ \idty & 0 \end{array} \right) = \mathcal{J}^t.
 \]

Here and in the following $A^t$ is the (non-hermitean) transpose of a (not necessarily square) matrix $A$, i.e.\ its matrix-elements are $(A^t)_{jk} = A_{kj}$. We will use this notation for matrices $A$ with scalar entries as well as for the case of operator-valued entries. Technically, $\mathcal{C}$, $\mathcal{J}$ as well as many of the quantities introduced below depend on $n$, just as the Hamiltonian $H_n$, but we will drop this from our notation to avoid over-subscripting.

We find it useful to express the xy Hamiltonian $H_n$ in terms of $\mathcal{C}$. Symmetrizing (\ref{eq:hamcs}) using (\ref{eq:car}),
it is clear that
\begin{eqnarray} \label{eq:Mc}
H_n & = & 2 \sum_{j=1}^{n-1} \mu_j[-c^*_j c_{j+1} - c_{j+1}^* c_j + \gamma_j (c_j c_{j+1} + c_{j+1}^* c_j^*)] + \sum_{j=1}^n\nu_j (2 c_j^* c_j- \idty) \nonumber \\
& = & \sum_{j=1}^{n-1} \mu_j \left[-c^*_j c_{j+1} + c_{j+1}c^*_j - c_{j+1}^* c_j + c_j c_{j+1}^* \right] + \nonumber \\
& \mbox{ } & \quad + \sum_{j=1}^{n-1} \mu_j  \gamma_j \left[ c_j c_{j+1} -c_{j+1}c_j+ c_{j+1}^* c_j^* -c_j^* c_{j+1}^* \right] + \sum_{j=1}^n\nu_j (c_j^* c_j - c_jc_j^*) \nonumber \\
& = & \mathcal{C}^* M \mathcal{C} \, ,
\end{eqnarray}
where $M$ is the $2\times 2$-block matrix
\begin{equation} \label{eq:Mmatrix}
M = \left( \begin{array}{cc} A & B \\ -B & -A \end{array} \right)
\end{equation}
with $A$ and $B$ Jacobi matrices satisfying
\begin{equation} \label{eq:A}
A = \left( \begin{array}{ccccc} \nu_1 & - \mu_1 & & & \\ - \mu_1 & \ddots & \ddots & & \\ & \ddots & \ddots & \ddots & \\ & & \ddots & \ddots & - \mu_{n-1} \\ & & & - \mu_{n-1} & \nu_n \end{array} \right)
\end{equation}
and
\begin{equation} \label{eq:B}
\quad B = \left( \begin{array}{ccccc} 0 &- \mu_1 \gamma_1 & & & \\  \mu_1 \gamma_1& \ddots & \ddots& & \\ & \ddots & \ddots & \ddots & \\ & & \ddots & \ddots &- \mu_{n-1} \gamma_{n-1} \\ & & & \mu_{n-1} \gamma_{n-1} & 0 \end{array} \right).
\end{equation}
Observe that $A^* = A^t = A$, $B^* = B^t = -B$, and thus $M^* = M^t = M$.

We complete the diagonalization of $H_n$ by applying a Bogoliubov-type transformation. To see this, set $S=A+B$ and
denote by $0\le \lambda_1 \le \lambda_2 \le \ldots \le \lambda_n$ the singular values of
$S$ counted with multiplicity.
The singular value decomposition of $S$ gives real-valued orthogonal matrices $U$ and $V$ such that
\begin{equation} \label{eq:svd}
USV^t = U(A+B)V^t = \Lambda,
\end{equation}
where $\Lambda = \mbox{diag}(\lambda_1, \ldots, \lambda_n)$.  In this case,
\[
\Lambda = \Lambda^t = VS^t U^t = V(A-B)U^t.
\]
Set
\[
W = \frac{1}{2} \left( \begin{array}{cc} V+U & V-U \\ V-U & V+U \end{array} \right) \, ,
\]
and note that $W$ is an orthogonal matrix.  This can be checked by directly verifying that $WW^t=\idty$, or, alternatively, by noting that
\[
SWS^{-1} = \left( \begin{array}{cc} V & 0 \\ 0 & U \end{array} \right) \quad \mbox{with} \quad S = \frac{1}{\sqrt{2}} \left( \begin{array}{cc} \idty & \idty \\ -\idty & \idty \end{array} \right) \, ,
\]
an orthogonal matrix. A short calculation shows that $W$ diagonalizes $M$, i.e.,
\begin{equation} \label{eq:diagtildeM}
W \left( \begin{array}{cc} A & B \\ -B & -A \end{array} \right) W^t = \left( \begin{array}{cc} \Lambda & 0 \\ 0 & -\Lambda \end{array} \right).
\end{equation}

Finally, we define $b$-operators by setting
\begin{equation} \label{eq:defB}
\mathcal{B} = W\mathcal{C} \, .
\end{equation}
It is easy to check that $\mathcal{B}$ has the form
\[
\mathcal{B} = (b_1, \ldots, b_n, b_1^*, \ldots, b_n^*)^t.
\]
Moreover, since $[\mathcal{J}, W] = 0$, it is clear that
\[
\mathcal{B}\mathcal{B}^* + \mathcal{J} (\mathcal{B}\mathcal{B}^*)^t \mathcal{J} = W \left( \mathcal{C} \mathcal{C}^* + \mathcal{J} (\mathcal{C}\mathcal{C}^*)^t \mathcal{J}  \right) W^t = \idty \, ,
\]
i.e., the $b$-operators satisfy CAR as well.

{F}rom (\ref{eq:Mc}), we conclude that
\begin{eqnarray} \label{eq:anisotropicFermi}
H_n  =  \mathcal{C}^* M \mathcal{C} & = & \mathcal{B}^* \left( \begin{array}{cc} \Lambda & 0 \\ 0 & - \Lambda \end{array} \right) \mathcal{B} \nonumber \\
& = &  \sum_{j=1}^n \lambda_j (b_j^* b_j - b_j b_j^*) \nonumber \\
& = & 2 \sum_{j=1}^n \lambda_j b_j^* b_j - E^{(n)}  \idty,
\end{eqnarray}
where $E^{(n)} = \sum_{j=1}^n \lambda_j$.

This means that when written in terms of the $b$-operators the xy Hamiltonian takes the form of a free Fermion system. This allows to explicitly diagonalize $H_n$ (in as far as the singular value decomposition (\ref{eq:svd}) is explicit), which we will exploit further in the next section for the case of the isotropic xy chain.

We close this section by observing that (\ref{eq:anisotropicFermi}) makes it easy to compute the time evolution of the $b$-operators. In fact, for any $1 \leq j \leq n$,
\[
\tau_t^n(b_j) = e^{-2it\lambda_j} b_j \quad \mbox{and} \quad \tau_t^n(b_j^*) = e^{2it\lambda_j} b_j^*,
\]
where $\tau_t^n$ is the Heisenberg dynamics corresponding to $H_n$. This follows, for example, by observing that, due to (\ref{eq:anisotropicFermi}) and CAR, $[H_n, b_j] = -2\lambda_j b_j$ and thus
\[
\frac{d}{dt} \tau_t^n(b_j) = i \tau_t^n ([H_n,b_j]) = -2i \lambda_j \tau_t^n(b_j) \,  \quad \mbox{and} \quad \tau_0^n(b_j) = b_j \, .
\]
In vector form, this can be written as
\[
\tau_t^n({\mathcal B}) = \left( \begin{array}{cc} e^{-2it\Lambda} & 0 \\ 0 & e^{2it\Lambda} \end{array} \right) {\mathcal B}.
\]
More is true. Since $W$ is a matrix of scalars, it is clear that
\[
\tau_t^n({\mathcal C}) = W^t \tau_t^n( \mathcal{B}) = e^{-2itM} \mathcal{C} \, ,
\]
or, componentwise,
\begin{equation} \label{eq:expandcj}
\tau_t^n(c_j) = \sum_{k=1}^n M_{j,k}(2t) c_{k}  + \sum_{k =1}^n M_{j,n+k}(2t) c_{k}^*
\end{equation}
where
\[
M_{j,k}(t) = \left( e^{-iM t} \right)_{j,k}.
\]
The above facts will be useful in the following sections. 

\vspace{.3cm}

\noindent {\bf Remark:} The specific structure of $M$  suggests the change of basis $(e_1, e_2, \ldots, e_{2n})$ to\\ $(e_1, e_{n+1}, e_2, e_{n+2}, \ldots, e_n, e_{2n})$ in representing $M$. This leads to
\begin{equation} \label{eq:gentb}
M \cong \tilde{M} :=\left( \begin{array}{ccccc} \nu_1 J & -\mu_1 S(\gamma_1) & & & \\ -\mu_1 S(\gamma_1)^t & \nu_2 J & \ddots & & \\ & \ddots & \ddots & \ddots & \\ & & \ddots & \ddots & -\mu_{n-1} S(\gamma_{n-1}) \\ & & & -\mu_{n-1} S(\gamma_{n-1})^t & \nu_n J \end{array} \right).
\end{equation}
Here
\[
J := \left( \begin{array}{cc} 1 & 0 \\ 0 & -1 \end{array} \right), \quad S(\gamma) := \left( \begin{array}{cc} 1 & \gamma \\ -\gamma & -1 \end{array} \right).
\]

In this view $M$ takes the form of a generalized tight-binding Hamiltonian with a (sign-indefinite) potential term generated by the magnetic field $(\nu_j)$ and non-standard hopping terms created by the $2\times 2$-matrix-valued off-diagonal terms $-\mu_j S(\gamma_j)$ and their adjoints. Mathematically, (\ref{eq:gentb}) provides the possibility to investigate spectral properties of $M$ with a transfer matrix formalism, although of higher order than in the case of standard tri-diagonal Jacobi matrices.

%
%
%
%
%

\subsection{From Dynamical Localization to Zero Velocity Bounds} \label{subsec:dl0v}

%
%

In the previous subsection, we demonstrated how to diagonalize a general
xy Hamiltonian. The goal of this subsection is to prove that dynamical localization of the random matrix $M= M^{(n)}$ associated with $H_n$ via (\ref{eq:Mc}) implies a zero-velocity Lieb-Robinson bound for $H_n$.

We begin by discussing what it means for $M^{(n)}$ to be dynamically localized.
Let the infinite sequences $\{ \mu_j \}$, $\{ \gamma_j \}$, and $\{ \nu_j \}$, which define an
xy model and thus the matrices $M^{(n)}$, see Section~\ref{subsec:diag}, correspond to sequences of random variables over a probability space $(\Omega, {\mathcal M}, \PP)$;
some concrete examples are provided in the following section. Thus the $M^{(n)}$ are self-adjoint random matrices and we will define dynamical localization in terms of bounds on the expected values of the matrix elements of their time evolution. In the following $\E(X) = \int_{\Omega} X \,d\PP$ denotes the expectation of a random variable $X$ on $\Omega$ with respect to $\PP$.

\begin{definition} \label{def:dynloc} We say the matrices $M^{(n)}$ are dynamically localized
if there exist numbers $C >0$ and $\eta >0$ such that for any integers $j,k, n \geq 1$
with $j,k \in [1,n]$,
\begin{equation} \label{eq:dynloc}
\mathbb{E} \left( \sup_{t\in \R} |M_{j, k}^{(n)}(t)| + \sup_{t\in \R} |M_{j, n+k}^{(n)}(t)|  \right) \le C e^{-\eta |j-k|} \, .
\end{equation}
Here $M_{j,k}^{(n)}(t) = \left( e^{-iM^{(n)} t} \right)_{j,k}$, where $M^{(n)}$ is the matrix associated to $H_n$
via (\ref{eq:Mc}).
\end{definition}

Dynamical localization can be equivalently expressed in terms of the rearranged matrix $\tilde{M} = \tilde{M}^{(n)}$ from (\ref{eq:gentb}),
\[
\E\left( \sup_{t\in \R} \left| (e^{-i\tilde{M}^{(n)}t})_{j,k} \right| \right) \le C' e^{-\eta' |j-k|}
\]
for all $n$ and $j,k \in [1,2n]$ with slightly modified constants $C'$ and $\eta'$.

It is important to note that the numbers $C>0$ and $\eta>0$ in Definition~\ref{def:dynloc}
are independent of $n$. In this case, the bound in (\ref{eq:dynloc}) is uniform with respect to
arbitrary finite volumes. The next section describes a class of xy models which
can be shown to satisfy Definition~\ref{def:dynloc}.

The main result of this subsection is that dynamical localization of $M^{(n)}$ implies a zero-velocity Lieb-Robinson bound after averages over the disorder are taken.

\begin{theorem} \label{thm:main}
Assume that $M^{(n)}$ is dynamically localized in the sense of Definition~\ref{def:dynloc}.
There are numbers $C'>0$ and $\eta>0$ such that for any integers $1 \leq j <k$ and
any $n \geq k$, the bound
\begin{equation} \label{eq:transbound}
\E \left(\sup_{t \in \mathbb{R}} \left\| \left[ \tau_t^n(A), B \right] \right\|  \right) \le C' \|A\| \|B\| e^{-\eta |k-j|}
\end{equation}
holds for all $A \in \mathcal{A}_j$ and $B \in \mathcal{A}_{ [k,n]}$.

In fact, with $C$ and $\eta$ from (\ref{eq:dynloc}), we may choose the same exponent $\eta$ and
\begin{equation} \label{eq:exptransbound}
C' = \frac{96C}{(1-e^{-\eta})^2}.
\end{equation}
\end{theorem}

\begin{proof}
Fix $n \geq k$. Let us first examine the case that $A = c_j$. Using (\ref{eq:expandcj}), it is clear that
\[
\left[ \tau_t^n(c_j), B \right]  =  \sum_{k'=k}^n M_{j,k'}(2t) \left[ c_{k'}, B \right]  + \sum_{k' =k}^n M_{j,n+k'}(2t) \left[ c_{k'}^*, B \right] \, ,
\]
where we have used that $B \in \mathcal{A}_{[k,n]}$, $c_{k'}, c_{k'}^* \in \mathcal{A}_{[1,k']}$ and thus $[c_{k'},B] = [c_{k'}^*,B]=0$ for $k'<k$. The bound (\ref{eq:dynloc}) implies that
\begin{equation} \label{eq:lrbcs}
\mathbb{E} \left( \sup_t \left\| \left[ \tau_t^n(c_j), B \right] \right\|  \right) \leq 4 C \| B \| \sum_{k' = k}^n e^{- \eta(k'-j)} \leq \frac{4C \| B \|}{1-e^{- \eta}}e^{- \eta(k-j)} \, ,
\end{equation}
where we have used that $\| [ A, B ] \| \leq 2 \| A \| \| B \|$ holds in general. By taking adjoints, the case of $A = c_j^*$
follows with an identical estimate.

Now suppose that $A= a_j$. By (\ref{eq:cdef}) we have $a_j = \sigma_1^z \ldots \sigma_{j-1}^z c_j$. Using the automorphism property of $\tau_t^n$ and the Leibnitz rule, we have that
\[
[\tau_t^n(a_j), B ] = \tau_t^n(\sigma_1^z) \cdots \tau_t^n(\sigma_{j-1}^z) [\tau_t^n(c_j), B] + [\tau_t^n(\sigma_1^z) \cdots \tau_t^n(\sigma_{j-1}^z), B ] \tau_t^n(c_j).
\]
Clearly then
\begin{equation} \label{eq:abd1}
\mathbb{E} \left( \sup_t \left\| \left[ \tau_t^n(a_j), B \right] \right\|  \right) \leq \frac{4C \| B \|}{1-e^{- \eta}}e^{- \eta(k-j)} + \mathbb{E} \left( \sup_t \left\| \left[ \tau_t^n(\sigma_1^z) \cdots \tau_t^n(\sigma_{j-1}^z), B \right] \right\|  \right) \, ,
\end{equation}
where we have used (\ref{eq:lrbcs}). We now further expand the second term above.

For any $1 < \ell < k$, it is convenient to define
\[
C(\ell, B) = \mathbb{E} \left( \sup_t \left\| \left[ \tau_t^n(\sigma_1^z) \cdots \tau_t^n(\sigma_{\ell}^z), B \right] \right\|  \right) \, .
\]
We claim that
\begin{equation} \label{eq:cit}
C(\ell, B) \leq C(\ell-1, B) +  \frac{4^2 C \| B \|}{1-e^{- \eta}}e^{- \eta(k- \ell)} \, .
\end{equation}
To see this, note that another application of Leibnitz shows
\begin{equation} \label{eq:leib}
\left\| \left[ \tau_t^n(\sigma_1^z) \cdots \tau_t^n(\sigma_{\ell}^z), B \right] \right\| \leq \left\| \left[\tau_t^n( \sigma_{\ell}^z), B \right] \right\| + \left\| \left[ \tau_t^n(\sigma_1^z) \cdots \tau_t^n(\sigma_{\ell-1}^z), B \right] \right\|  \, .
\end{equation}
Moreover, for any $\ell$,
\[
\sigma_{\ell}^z = 2a_{\ell}^* a_{\ell} - \idty = 2 c_{\ell}^*c_{\ell} - \idty \, ,
\]
and so
\begin{equation} \label{eq:cprod}
\left[ \tau_t^n( \sigma_{\ell}^z), B \right] = 2 \left[ \tau_t^n( c_{\ell}^*), B \right] \tau_t^n(c_{\ell}) + 2 \tau_t^n(c_{\ell}^*) \left[ \tau_t^n( c_{\ell}), B \right] \, .
\end{equation}
It is now clear that (\ref{eq:cit}) follows from (\ref{eq:leib}), (\ref{eq:cprod}), and (\ref{eq:lrbcs}). Additionally, the quantity corresponding
to $\ell =1$ satisfies
\[
C(1,B) \leq \frac{4^2 C \| B \|}{1-e^{- \eta}}e^{- \eta(k- 1)} \, ,
\]
as is clear from (\ref{eq:cprod}).

As a result, our estimate now follows by iteration. From (\ref{eq:abd1}), we have that
\begin{eqnarray*}
\mathbb{E} \left( \sup_t \left\| \left[ \tau_t^n(a_j), B \right] \right\|  \right) & \leq & \frac{4C \| B \|}{1-e^{- \eta}}e^{- \eta(k-j)} + C(j-1, B) \nonumber \\
& \leq & \frac{4C \| B \|}{1-e^{- \eta}} \left( e^{- \eta(k-j)} + 4 \sum_{\ell = 1}^{j-1} e^{- \eta(k-\ell)} \right) \nonumber \\
& \leq & \frac{4^2C \| B \|}{1-e^{- \eta}} e^{- \eta(k-j)} \left( 1 +  \sum_{\ell = 1}^{j-1} e^{- \eta(j-\ell)} \right) \nonumber \\
& \leq & \frac{4^2C \| B \|}{(1-e^{- \eta})^2} e^{- \eta(k-j)} \, .
\end{eqnarray*}

Again, the case of $A = a_j^*$ follows, with the same estimate as above, by taking adjoints.
Using Leibnitz, the above bound yields estimates for both $A = a_j^*a_j$ and $A= a_ja_j^*$,
increasing the bound only by $2$. Since the collection $ \{ a_j, a_j^*, a_j^*a_j, a_ja_j^*\}$
constitutes the canonical basis for $\mathcal{A}_j$, we get (\ref{eq:transbound}) with the explicit constants (\ref{eq:exptransbound}), using the simple fact that the norm of a matrix bounds the absolute values of all of its entries. A slightly more refined argument would allow to lower the constant $96$ in (\ref{eq:exptransbound}) a bit.
\end{proof}

A number of comments about Theorem~\ref{thm:main} are in order.

\begin{remark}
Theorem~\ref{thm:main} applies to more general observables. For example, let $J \subset \mathbb{Z}$ be a finite
set.  To each $j \in J$ and any $1 \leq n,m \leq 2$, denote by $e_j(n,m) \in \mathcal{A}_j$ a matrix unit, i.e., a matrix with $1$ in
the $(n,m)$ entry and all other entries equal to $0$. Any $A \in \mathcal{A}_J$
can be expanded in terms of products of matrix units; there will be $2^{2|J|}$ such terms. Here $|J|$ denotes
the cardinality of $J$. Using Leibnitz to expand the product, it is clear that
\begin{equation} \label{eq:moregenobs}
\E \left(\sup_{t \in \mathbb{R}} \left\| \left[ \tau_t^n(A), B \right] \right\|  \right) \le  2^{2|J|} \tilde{C} \|A\| \|B\| e^{-\eta d(k,J)}
\end{equation}
holds for all $A \in \mathcal{A}_J$ and $B \in \mathcal{A}_{ [k,n]}$ if $\max[j: j \in J] < k \leq n$. Here $d(k,J) = \min[|k-j| : j \in J]$.
\end{remark}

The bound (\ref{eq:moregenobs}) is useful in applications where $J$ is kept fixed as $n$ varies. However, when applying (\ref{eq:moregenobs}) to cases such as $A\in \mathcal{A}_{[1,j]}$, $B\in \mathcal{A}_{[k,n]}$, $j<k$, the constant on the right grows exponentially in $j$, while one would hope to have bounds exponentially small in $|k-j|$ and constants uniform in $n$ and $1\le j<k\le n$.

The next result, which we state as a corollary, achieves a bound of this type, holding for a larger class of observables than Theorem~\ref{thm:main}. Moreover, this result provides better bounds for small times, e.g.\ $|t|<1$. The drawback of this result, and the method to prove it, is that one pays a penalty for large times.

\begin{corollary} \label{cor:smalltimeLR}
Consider an xy model that is dynamically localized in the sense
of Definition~\ref{def:dynloc}. If all of the random variables are uniformly bounded, then
there exist numbers $c>0$ and $\eta>0$ such that for any finite set $J \subset \N$ and any integers $n, k$ satisfying $\max \{j:j\in J\} < k \le n$ the estimate
\[
\mathbb{E} \left( \left\| \left[ \tau_t^n(A), B \right] \right\| \right) \leq c |t| \| A \| \| B \| e^{- \eta |k-j|}
\]
holds for all $A \in \mathcal{A}_J$ and $B \in \mathcal{A}_{[k,n]}$.
\end{corollary}

\begin{proof}
Fix $A \in \mathcal{A}_J$ and $B \in \mathcal{A}_{[k,n]}$.
Consider the function
\[
f(t) = \left[ \tau_t^n(A), B \right] \, .
\]
Here $f$ depends on $n$, $A$, and $B$, but we will suppress this in our notation. Calculate:
\begin{eqnarray*}
f'(t) & = & i \left[ \tau_t^n \left( [H_n,A] \right), B \right] \nonumber \\
& = &  i \left[ \tau_t^n \left( [H_J,A] \right), B \right] \nonumber \\
& = & i \left[  [ \tau_t^n(H_J), \tau_t^n(A)] , B \right] \nonumber \\
& = & -i \left[  [ \tau_t^n(A), B ],  \tau_t^n(H_J) \right]  -i \left[  [ B, \tau_t^n(H_J)] , \tau_t^n(A) \right] \, .
\end{eqnarray*}
Above we have denoted by $H_J$ those terms in $H_n$ that have non-trivial commutator with $A$, and the final
equality follows from the Jacobi identity.
This proves that $f$ satisfies the differential equation:
\[
f'(t) = - i \left[ f(t), \tau_t^n(H_J) \right] -i \left[  [ B, \tau_t^n(H_J)] , \tau_t^n(A) \right] \, .
\]
Since the first term is norm-preserving and $f(0)=0$, we get
\begin{equation} \label{eq:normpresbd}
\left\| \left[ \tau_t^n(A), B \right] \right\| = \| f(t) \| \leq 2 \| A \| \, \int_0^{|t|} \left\| \left[ \tau_s^n(H_J), B \right] \right\| \, ds \, ,
\end{equation}
see e.g.\ Appendix A of \cite{nach22006}.

Given an xy model of the form (\ref{eq:anisoxychain}), it is clear that
\[
H_J = \sum_{j \in \tilde{J}} \mu_j [ (1+\gamma_j) \sigma_j^x \sigma_{j+1}^x + (1-\gamma_j) \sigma_j^y \sigma_{j+1}^y] + \sum_{j \in J} \nu_j \sigma_j^z \, ,
\]
where $\tilde{J} = [J \cup (J-1)] \cap [1,n]$. For any $w \in \{x,y\}$, the estimate
\[
\mathbb{E} \left( \left\| \left[ \tau_s^n( \sigma_j^w \sigma_{j+1}^w), B \right] \right\| \right) \leq C' (1+ e^{\eta}) \| B \| e^{- \eta(k-j)} \, ,
\]
follows from (\ref{eq:transbound}) by Leibnitz; it holds for each $j \in \tilde{J}$ and is uniform in $s \geq 0$.
Using (\ref{eq:normpresbd}) then, we find that
\begin{eqnarray*}
\mathbb{E} \left( \left\| \left[ \tau_t^n(A), B \right] \right\| \right) & \leq &  2 \| A \| \, \int_0^{|t|} \mathbb{E} \left( \left\| \left[ \tau_s^n(H_J), B \right] \right\| \right) \, ds \, \nonumber \\
& \leq & 4 \tilde{C} \| A \| \sum_{j \in \tilde{J}} C'(1 + e^{\eta}) \| B \| e^{- \eta(k-j)} |t| + \nonumber \\
& \mbox{ } & \quad + 2 \hat{C} \| A \| \sum_{j \in J} C' \| B \| e^{- \eta(k-j)} |t| \, \nonumber \\
& \leq & c |t| \| A \| \| B \| e^{- \eta d(k,J)} \, ,
\end{eqnarray*}
where $\tilde{C} = \max[ |\mu_j|(1+|\gamma_j|)]$, $\hat{C} = \max[|\nu_j|]$ and $c:= (4\tilde{C}C'(1+e^{\eta}) +2\hat{C} C')/(1-e^{-\eta})$.
\end{proof}


\section{The isotropic random xy chain} \label{sec:isoxychain}

In this section we will apply the above results to the isotropic xy chain in an exterior random magnetic field. This corresponds to the special case $\gamma_j =0$ and $\mu_j = \mu \not=0$ for all $j$ in (\ref{eq:anisoxychain}), giving the Hamiltonian
\begin{equation} \label{eq:isoxychain}
H_{n,{\rm iso}} = \mu \sum_{j=1}^n [ \sigma_j^x \sigma_{j+1}^x + \sigma_j^y \sigma_{j+1}^y] + \sum_{j=1}^n \nu_j \sigma_j^z.
\end{equation}
We will assume that

\vspace{.3cm}

\noindent {\bf (A)} The coupling parameters $\nu_j$, $j\in \N$, are i.i.d.\ random variables, whose common distribution $\PP_0$ is absolutely continuous with bounded and compactly supported density $\rho$.

\vspace{.3cm}

In this case one may choose $(\Omega,{\mathcal M}, \PP)$ as the infinite product space $\bigotimes_{j\in \N} (\R, {\mathcal B}, \PP_0)$ with the Borel algebra ${\mathcal B}$.

\subsection{Dynamical Localization} \label{sec:isoxydynloc}

For the isotropic xy chain the matrix $M$ in (\ref{eq:Mmatrix}) becomes block-diagonal,
\begin{equation} \label{eq:isoM}
M = \left( \begin{array}{cc} A & 0 \\ 0 & -A \end{array} \right),
\end{equation}
where
\[
A = \left( \begin{array}{ccccc} \nu_1 & - \mu & & & \\ - \mu & \ddots & \ddots & & \\ & \ddots & \ddots & \ddots & \\ & & \ddots & \ddots & - \mu \\ & & & - \mu & \nu_n \end{array} \right)
\]
is the restriction of the classical Anderson model to $[1,n]$. Dynamical localization (\ref{eq:dynloc}) for $M$ becomes equivalent to dynamical localization for the Anderson model in the form
\begin{equation} \label{eq:dynlocAnderson}
\E \left( \sup_{t\in \R} |(e^{-iAt})_{j,k}| \right) \le C e^{-\eta |j-k|}
\end{equation}
for constants $C>0$ and $\eta>0$ uniformly in $n\in \N$ and $1\le j,k \le n$. Under the above assumption on the distribution of $\nu_j$ this is well known. It has been proven by two different methods:

(a) The Kunz-Souillard method goes back to \cite{KunzSouillard}, see \cite{CFKS} and \cite{DamanikSurvey} for detailed presentations. We remark here that these sources provide a proof of (\ref{eq:dynlocAnderson}) for the Anderson model over $[-n,n]$ (instead of $[1,n]$) and for $(j,k) = (0,k)$, $0<k\le n$. The method can easily be adjusted to handle general pairs $(j,k)$.

(b) The fractional moments method was initially developed by Aizenman and Molchanov in \cite{AM} and is first shown to yield dynamical localization in \cite{Aizenman}. Detailed arguments leading to (\ref{eq:dynlocAnderson}) can be found by combining the exposition in Section~6 of \cite{StolzArizona} with fractional moment bounds for the one-dimensional Green function as provided by Proposition~A.1 in \cite{Minami} or Theorem~4.1 in \cite{HSS1}.

In summary, we have the following corollary of Theorem~\ref{thm:main}.
\begin{corollary} \label{cor:0vLRisoxy}
Suppose that the magnetic field in the isotropic xy chain (\ref{eq:isoxychain}) satisfies {\bf (A)}. Then the Heisenberg evolution
\[ \tau_t^n(A) = e^{i H_{n,{\rm iso}} t} A e^{-i H_{n,{\rm iso}} t}\]
satisfies (\ref{eq:transbound}), a zero velocity Lieb-Robinson bound in disorder average.
\end{corollary}

\vspace{.3cm}

\noindent {\bf Remark:} Corollary~\ref{cor:0vLRisoxy} improves on a result stated in \cite{BO}, where it was argued that the isotropic xy-chain in random magnetic field satisfies the bound
\begin{equation} \label{eq:BO}
\| [\tau_t^n(A),B]\| \le c n^2 |t| e^{-\eta |j-k|}
\end{equation}
for some $\eta>0$ and $A\in {\mathcal A}_j$, $B\in {\mathcal A}_k$. This was interpreted as a {\it logarithmic light cone} for information propagation in the spin chain. As in our approach, \cite{BO} uses the Jordan-Wigner transform to reduce this to localization properties for the Anderson model. In the present work we use suitable rigorous versions of the latter to show that, after disorder averaging on the left hand side of (\ref{eq:BO}), one gets strict dynamical localization in the sense that the right hand side neither grows in time nor in the size $n$ of the chain.

We will not attempt here to survey other results on disordered spin systems in the physics literature, but we mention several more works on closely related models: Numerical results on the dynamics as well as on spin-spin correlations of the Heisenberg XXZ chain in random magnetic field are presented in \cite{ZPP}. In \cite{BEO} quantum spin chains with temporally fluctuating exterior fields are studied and it is discussed that these lead to dynamical phenomena different from those observed for models with static disorder. Finally, a very comprehensive discussion of the physics of different types of spin 1/2 chains (in particular XY and XXZ) with random exchange couplings (e.g.\ the parameters $\mu_j$ and $\eta_j$ in the case of the XY model) can be found in \cite{Fisher}.

\subsection{Decay of ground state correlations} \label{subsec:isocordecay}

Corollary~\ref{cor:0vLRisoxy} allows, at least in principle, to apply Theorem~\ref{thm:expdec} to the isotropic xy chain in random magnetic field and conclude exponential decay of ground state correlations. This is not completely straightforward as there are two additional issues which need to be addressed:

(i) Theorem~\ref{thm:expdec} is deterministic in nature, saying that a zero-velocity Lieb-Robinson bound for a fixed quantum spin Hamiltonian implies exponential decay of ground state correlations for this Hamiltonian. Corollary~\ref{cor:0vLRisoxy}, however, provides the zero-velocity LR bound in {\it expectation}. To get a probabilistic version of correlation decay from this we will have to slightly adapt the proof of Theorem~\ref{thm:expdec}.

(ii) We need information on the gap size $\gamma$ for the random xy chain to control the logarithmic correction in (\ref{eq:expdec}). In doing so we will have to take into account that $\gamma$ itself is a random variable. To achieve the second goal, we start by providing some more background on the explicit diagonalization of the free Fermion system (\ref{eq:anisotropicFermi}).

As we are interested in the isotropic xy chain and thus $M$ as in (\ref{eq:isoM}), we can simply choose
\[ W = \left( \begin{array}{cc} U & 0 \\ 0 & U \end{array} \right) \]
in (\ref{eq:diagtildeM}) where $U$ is orthogonal and diagonalizes the Anderson model,
\[ U A U^t = \Lambda = \mbox{diag}(\lambda_1, \ldots, \lambda_n).\]
Note that this choice is slightly different from what was done in Section~\ref{subsec:diag}, as $\lambda_1 \le \ldots \le \lambda_n$ are now the eigenvalues of the Anderson model, which are not necessarily non-negative. Proceeding very similar to above, in fact a bit simpler, we define
\[ \left( \begin{array}{c} b_1 \\ \vdots \\ b_n \end{array} \right) := U \left( \begin{array}{c} c_1 \\ \vdots \\ c_n \end{array} \right),\]
see that $b_j$, $j=1,\ldots,n$, satisfy the CAR and allow to represent the isotropic xy chain as the free Fermion system
\[ H_{n,{\rm iso}} = 2 \sum_{j=1}^n \lambda_j b_j^* b_j - E^{(n)} \idty,\]
where, as before, $E^{(n)} = \sum_{j=1}^n \lambda_j$. The $b_j^*$ and $b_j$ are interpreted as creation and annihilation operators. The following facts are well known, see e.g.\ \cite{Simon}.

\begin{itemize}

\item The operators $b_j^* b_j$, $j=1,\ldots,n$, are pairwise commuting orthogonal projections, meaning that they can be simultaneously diagonalized.

\item The intersection of the kernels of $b_j^* b_j$ is one-dimensional. Thus it is spanned by an essentially unique normalized vector $\Omega$ (the {\it vacuum vector}), so that  $b_j^*b_j \Omega = 0$, $j=1,\ldots,n$.

\item For each $\alpha = (\alpha_1, \ldots, \alpha_n) \in \{0,1\}^n$ use successive creation operators to define
\[
\psi_{\alpha} = (b^*)^{\alpha} \Omega := (b_1^*)^{\alpha_1} \ldots (b_n^*)^{\alpha_n} \Omega,
\]
which are an orthonormal basis of ${\mathcal H}_{[1,n]}$ consisting of common eigenvectors for the $b_j^* b_j$:
\begin{equation} \label{eq:evecprop}
b_j^* b_j \psi_{\alpha} = \left\{ \begin{array}{ll} 0, & \mbox{if $\alpha_j =0$}, \\ \psi_{\alpha}, & \mbox{if $\alpha_j=1$}. \end{array} \right.
\end{equation}

\end{itemize}

Clearly, the $\psi_{\alpha}$ are also eigenvectors for $H_{n,{\rm iso}}$,
\[ H_{n,{\rm iso}} \psi_{\alpha} = \left( 2 \sum_{j:\alpha_j=1} \lambda_j -E^{(n)} \right) \psi_{\alpha} = \left( \sum_{j:\alpha_j=1} \lambda_j - \sum_{j:\alpha_j=0} \lambda_j \right) \psi_{\alpha}.\]

This shows that the ground state energy of $H_n$ is $E_0 = - \sum_{j=1}^n |\lambda_j|$ and that $E_0$ is non-degenerate if and only if $\lambda_j \not= 0$ for all $j$. In this case the ground state is $\psi_{\alpha^{(0)}}$, where
\begin{equation} \label{eq:groundstate}
\alpha_j^{(0)} := \left\{ \begin{array}{ll} 1, & \mbox{if $\lambda_j<0$}, \\ 0, & \mbox{if $\lambda_j>0$}, \end{array} \right.
\end{equation}
and the gap $\gamma$ between the ground state energy $E_0$ and the first excited energy $E_1$ is
\begin{equation}  \label{eq:gapsize}
\gamma = E_1-E_0 = 2 \min_j |\lambda_j| = 2 \,\mbox{dist}(0,\sigma(A)).
\end{equation}
Thus probabilistic lower bounds for the gap size derive from the following Wegner estimate for the Anderson model, which holds under our assumptions, see e.g.\ Theorem~5.23 in \cite{Kirsch} for a detailed proof. It guarantees the existence of a constant $C_W<\infty$ such that
\begin{equation} \label{eq:Wegner}
\PP(\mbox{dist}(E,\sigma(A))<\varepsilon) \le C_W |\varepsilon| n
\end{equation}
uniformly in $n\in \N$, $\varepsilon >0$ and $E\in \R$. For $E=0$ this implies  by (\ref{eq:gapsize}) that
\begin{equation} \label{eq:goodguys}
\PP(\gamma \ge 2\varepsilon) \ge 1- C_W |\varepsilon| n.
\end{equation}
From (\ref{eq:Wegner}) we also see that $\PP(0\in \sigma(A)) =0$ and thus that the ground state of $H_{n,{\rm iso}}$ is almost surely non-degenerate. By $\psi_0$ we denote this almost sure essentially unique and normalized ground state (the choice of phase is irrelevant for the following result).

\begin{theorem} \label{thm:isoxycordecay}
Under assumption {\bf (A)} there exist $C<\infty$ and $\eta'>0$ such that
\begin{equation} \label{eq:isoxycordecay}
\E \left( |\langle \psi_0, AB \psi_0 \rangle - \langle \psi_0, A\psi_0 \rangle \langle \psi_0, B\psi_0 \rangle | \right) \le C\|A\| \|B\| \,n\, e^{-\eta' |j-k|}
\end{equation}
for all $A\in {\mathcal A}_j$, $B\in {\mathcal A}_{[k,n]}$ and all $1\le j<k\le n$.
\end{theorem}

\begin{proof}
We closely mimic the proof of Theorem~\ref{thm:expdec}, with the main difference being that the zero velocity Lieb-Robinson bound holds in expectation (i.e.\ in the form (\ref{eq:transbound})), and not deterministically as assumed in Theorem~\ref{thm:expdec}.

As the ground state of $H_{n,{\rm iso}}$ is almost surely non-degenerate we may assume by the remarks after Theorem~\ref{thm:expdec} that $\langle \psi_0, B\psi_0 \rangle =0$ and, after an energy shift, that $E_0=0$. We thus need to bound $\E(|\langle \psi_0, AB \psi_0 \rangle |)$.

Using Corollary~\ref{cor:0vLRisoxy} combined with an application of Corollary~\ref{cor:smalltimeLR} to $H_{n,{\rm iso}}$ yields the existence of $C>0$ and $\eta>0$ such that
\begin{equation} \label{eq:largesmall}
\E \left( \| [ \tau_t^n(A), B] \| \right) \le C \min[|t|,1] \|A\| \|B\| e^{-\eta |k-j|}
\end{equation}
for all $A\in {\mathcal A}_j$, $B\in {\mathcal A}_{[k,n]}$ and $1\le j< k\le n$.

Decompose $\Omega$ into good and bad events according to
\[
\Omega_g := \{ \nu = (\nu_j) \in \Omega: \gamma \ge e^{-\eta n}\}, \quad \Omega_b := \Omega \setminus \Omega_g.
\]
By (\ref{eq:goodguys}) we have
\begin{equation} \label{eq:badbound}
\PP(\Omega_b) \le \frac{1}{2}C_W \,n \, e^{-\eta n}.
\end{equation}
Accordingly we decompose
\begin{equation} \label{eq:goodbadsplit}
\E \left( | \langle \psi_0, AB \psi_0 \rangle | \right) = \E \left( \chi_{\Omega_g} \cdot |\langle \psi_0, AB \psi_0 \rangle | \right) + \E \left( \chi_{\Omega_b} \cdot |\langle \psi_0, AB \psi_0 \rangle | \right),
\end{equation}
where $\chi_{\Omega_g}$ and $\chi_{\Omega_b}$ denote the characteristic functions of $\Omega_g$ and $\Omega_b$.

With (\ref{eq:badbound}) we can bound the second term by
\begin{equation} \label{eq:sectermbound}
\E \left( \chi_{\Omega_b} \cdot |\langle \psi_0, AB \psi_0 \rangle | \right) \le \|A\| \|B\| \PP(\Omega_b) \le \frac{1}{2} C_W \|A\| \|B\| \,n \,e^{-\eta n} \le \frac{1}{2} C_W \|A\| \|B\| \,n \, e^{-\eta|k-j|}.
\end{equation}

The first term on the right hand side of (\ref{eq:goodbadsplit}) is dealt with in complete analogy with the proof of Theorem~\ref{thm:expdec}, starting with the three-way split in (\ref{eq:corr}). Due to the restriction to $\Omega_g$ we have $\gamma \ge e^{-\eta n}$. Thus the proper choice of $\alpha$ corresponding to (\ref{eq:alphalambda}) is $\alpha := e^{-2\eta n}/ (4 \eta |k-j|)$. We again set $\lambda := \sqrt{\pi/4\alpha}$.

Proceeding as in the proof of Theorem~\ref{thm:expdec} we get
\begin{equation} \label{eq:bd1}
\E \left( \chi_{\Omega_g} | \langle \psi_0, A(B-B(\alpha,\varepsilon)) \psi_0 \rangle | \right) \le \frac{1}{2}\|A\| \|B\| e^{-\eta|k-j|}
\end{equation}
and
\begin{equation} \label{eq:bd2}
\E \left( \chi_{\Omega_g} | \langle \psi_0, B(\alpha,\varepsilon) A \psi_0 \rangle | \right) \le \frac{1}{2}\|A\| \|B\| e^{-\eta|k-j|}.
\end{equation}
On the remaining term we use Fubini to bound, according to (\ref{eq:thirdterm}),
\begin{eqnarray*}
\E \left( \chi_{\Omega_g} | \langle \psi_0, [ A, B(\alpha,\varepsilon)] \psi_0 \rangle | \right) & \le & \E( \|[A, B(\alpha,\varepsilon)]\| \\
& \le & \frac{1}{2\pi} \int_{\R} \E \left( \frac{1}{|t|} \|[A, \tau_t^n(B)]\| \right) e^{-\alpha t^2}\,dt.
\end{eqnarray*}
As above, the $t$-integration is split into the the three regions $|t|\le 1$, $1\le |t| \le \lambda$ and $|t|\ge \lambda$. Proceeding as above, using (\ref{eq:largesmall}), one gets three bounds which combine to
\begin{equation} \label{eq:bd3}
\E \left( \chi_{\Omega_g} | \langle \psi_0, [ A, B(\alpha,\varepsilon)] \psi_0 \rangle | \right)  \le  C\|A\| \|B\| (1+ \ln \lambda) e^{-\eta|k-j|}.
\end{equation}
We finally observe that $\ln \lambda = \eta n + \frac{1}{2} \ln(\pi \eta |k-j|)$. The bounds (\ref{eq:bd1}), (\ref{eq:bd2}) and (\ref{eq:bd3}) combine into
\[ \E\left(\chi_{\Omega_g} \cdot | \langle \psi_0, AB \psi_0 \rangle | \right) \le C\|A\| \|B\| (1+ n+ \ln|k-j|) e^{-\eta|k-j|}.\]
Together with (\ref{eq:goodbadsplit}) and (\ref{eq:sectermbound}) this gives (\ref{eq:isoxycordecay}), where $\eta'$ can be any number less than $\eta$.

\end{proof}

\noindent {\bf Remark:} We note that Theorem~\ref{thm:isoxycordecay} is trivial for the case of strong magnetic field in the sense that either supp$\,\PP_0 \subset [2|\mu|,\infty)$ or supp$\,\PP_0 \subset (-\infty, -2|\mu|]$. In this case the ground state $\psi_0$ of $H_{n,{\rm iso}}$ has all spins down (or all spins up, respectively). Thus it is a product state with trivial correlations.

To see this, say for supp$\,\PP_0 \subset [2|\mu|,\infty)$, note that in this case all eigenvalues of the Anderson model $A$ are positive, $0<\lambda_1< \ldots < \lambda_n$. Thus, by (\ref{eq:groundstate}), the non-degenerate ground state of $H_{n,{\rm iso}}$ coincides with the vacuum $\Omega$ and
\[ E_0 = -\sum_{j=1}^n \lambda_j = - \mbox{tr}\,A = -\sum_{j=1}^n \nu_j.\]
But one also verifies directly that the all spins down product state is an eigenvector of $H_{n,{\rm iso}}$ to eigenvalue $-\sum_j \nu_j$.

Theorem~\ref{thm:isoxycordecay} is non-trivial for the case where $0$ lies in the bulk of the spectrum of the Anderson model and thus the ground state of $H_{n,{\rm iso}}$ becomes entangled.

\section{Concluding Remarks} \label{sec:conrem}

\subsection{The work by Klein and Perez \cite{KleinPerez}}

A result on exponential decay of ground state correlations for the isotropic $xy$-chain (\ref{eq:isoxychain}) closely related to Theorem~\ref{thm:isoxycordecay} was previously proven by Klein and Perez in \cite{KleinPerez}. They consider the special case of correlations between the raising and lowering operators $a_j^*$ and $a_k$ and show that there exists an $\eta>0$ such that
\begin{equation} \label{KPbound}
\sup_n | \langle \psi_0, a_j^* a_k \psi_0 \rangle| \le C_{\nu} e^{-\eta |j-k|}
\end{equation}
for all $j$ and $k$ and almost every choice of the magnetic field $\nu = (\nu_j)$. Their proof is also based on the Jordan-Wigner transform and then uses Wick's Theorem to expand the ground state correlations $\langle \psi_0, a_j^*, a_k \psi_0\rangle$ in terms of the Fermi two-point function $\langle \psi_0, c_j^* c_k \psi_0 \rangle$. The latter is given by matrix elements of a spectral projection of the underlying Anderson model. Exponential decay is concluded using localization results for the one-dimensional Anderson model found via the multiscale analysis approach. Compared to our approach this has the advantage of allowing for general distributions of the i.i.d.\ random variables $(\nu_j)$, including the most singular case of Bernoulli variables.

However, the results of \cite{KleinPerez} and our results on the $xy$ chain are not directly compatible in that Klein and Perez get correlation decay with probability one, while we get our decay bound after averaging over the disorder. In principle, expectation bounds such as (\ref{eq:isoxycordecay}) imply similar almost sure bounds with a reduced exponent by a standard summation argument. However, this would lead to a slight volume dependence in the almost sure bounds, due to the factor $n$ on the right hand side of (\ref{eq:isoxycordecay}) but also the fact that we work in finite volume and thus can't use translation invariance of the random variables considered here. Thus the almost sure bounds obtained from (\ref{eq:isoxycordecay}), while holding for larger classes of observables $A$, $B$, would be slightly weaker than the ones found in \cite{KleinPerez}. On the other hand, almost sure bounds such as (\ref{KPbound}) do not imply bounds on averages, due to the $\nu$-dependence of the constant in (\ref{KPbound}).

The most important new aspect of our approach, as compared to the direct proof of correlation decay in \cite{KleinPerez}, is that we prove dynamical localization in the form of Corollary~\ref{cor:0vLRisoxy} and then deduce correlation decay as a model independent consequence of this. This latter fact can be seen as somewhat reminiscent of a well-known result from the theory of random Schr\"odinger operators, saying that suitable forms of dynamical localization imply spectral localization (e.g.\ pure point spectrum with exponentially decaying eigenfunctions).

\subsection{More general random block operators}

Our strongest results on localization properties for random spin systems, e.g.\ Corollary~\ref{cor:0vLRisoxy} and Theorem~\ref{thm:isoxycordecay}, are restricted to the isotropic xy chain (\ref{eq:isoxychain}), as in this case we can readily refer to the underlying dynamical localization bound (\ref{eq:dynlocAnderson}) for the Anderson model. Extending these results to the anisotropic xy chain (\ref{eq:anisoxychain}) would require to prove dynamical localization (\ref{eq:dynloc}) for more general random block operators $M$ of the type (\ref{eq:Mmatrix}), (\ref{eq:A}), (\ref{eq:B}), under suitable assumptions on the random parameters $\nu_j$ and $\mu_j$. Even for the case of constant $\mu_j$ and i.i.d.\ $\nu_j$ this does not seem to be known. A difficulty which prevents a straightforward extension of existing methods is the non-monotonicity of $M$ in the random paramaters (the quadratic form of $M$ is not monotone in the $\nu_j$, as each appears twice, once with positive and once with negative sign). Also interesting is the case of random interaction strength $\mu_j$ of neighboring spins (and, for example, vanishing or constant magnetic field), leading to an Anderson-type operator $M$ with off-diagonal randomness. While the latter type of randomness has been studied in the literature, we are not aware of a result like (\ref{eq:dynloc}) in this case.

A further generalization of the anisotropic xy chain (\ref{eq:anisoxychain}) is found by allowing interactions of more than just next neighbors. By the same argument as in Section~\ref{subsec:diag} this can be reduced to diagonalizing a block operator $M$ as in (\ref{eq:Mmatrix}) with more general band matrices $A$ and $B$.

We also note that random operators of the slightly different form
\[ \left( \begin{array}{cc} A & B \\ B & -A \end{array} \right) \]
where both, $A$ and $B$, are self-adjoint have been studied in \cite{KMM}. Here the authors were motivated by physical applications in the BCS-theory of superconductivity. For these models of random block operators as well, questions relating to localization remain widely open and will require additional work.

%
%

\section*{Acknowledgements}

The authors express their gratitude to Bruno Nachtergaele for valuable discussions and to Jacob Chapman for a thorough proofreading of parts of this work. G.\ S.\ acknowledges hospitality at Centre Interfacultaire Bernoulli, Lausanne, Switzerland.
E.\ H.\ and R.\ S.\ both acknowledge the hospitality of the Erwin Schr\"odinger Institute in Vienna, Austria. 
E.\ H.\ is grateful for the hospitality and support of the University of Alabama at Birmingham during her visit.

\end{document}